\documentclass{llncs}
\usepackage{graphicx}
\usepackage{amssymb,boxedminipage,amsmath,comment}

\newcommand{\NP}{{\sf NP}}

\newcommand{\ssi}{\subseteq_i}
\newcommand{\si}{\supseteq_i}

\title{Critical Vertices and Edges in $H$-free Graphs\thanks{Results in this paper appeared in extended abstracts in the proceedings of ISCO 2016~\cite{PPR16} and LAGOS~2017~\cite{PPR17b}.}}

\author{Dani\"el Paulusma\inst{1}\thanks{Author supported by the Leverhulme Trust (RPG-2016-258).}
\and Christophe Picouleau\inst{2} \and Bernard Ries\inst{3}}

\institute{Durham University, Durham, UK, \texttt{daniel.paulusma@durham.ac.uk}
\and
CNAM, Laboratoire CEDRIC, Paris, France, \texttt{christophe.picouleau@cnam.fr}
\and
University of Fribourg, Department of Informatics, Fribourg, Switzerland,  \texttt{bernard.ries@unifr.ch}}

\begin{document}
\maketitle
\setcounter{footnote}{0}

\begin{abstract}
A vertex or edge in a graph is critical if its deletion reduces the chromatic number of the graph by~1. 
 We consider the problems of deciding whether a graph has a critical vertex or edge, respectively. 
We give a complexity dichotomy for both problems restricted to $H$-free graphs, that is, graphs with no induced subgraph isomorphic to $H$. Moreover, we show that an edge is critical if and only if its contraction reduces the chromatic number by~1.
Hence, we also obtain a complexity dichotomy for the problem of deciding if a graph has an edge whose contraction reduces 
the chromatic number by~1.

\medskip
\noindent
{\bf Keywords.} edge contraction, vertex deletion, chromatic number.
\end{abstract}

\section{Introduction}\label{s-intro}

For a positive integer $k$, a {\it $k$-colouring} of a graph $G=(V,E)$ is a mapping $c: V\rightarrow\{1,2,\ldots,k\}$ such that no two end-vertices of an edge are coloured alike, that is,  $c(u)\neq c(v)$ if $uv\in E$.  
The {\it chromatic number} $\chi(G)$ of a graph $G$ is the  smallest integer $k$ for which $G$ has a $k$-colouring. 
The well-known {\sc Colouring} problem is to test if $\chi(G)\leq k$ for a given graph $G$ and integer $k$.
If $k$ is not part of the input, then we call this problem $k$-{\sc Colouring} instead. Lov\'asz~\cite{Lo73} proved that $3$-{\sc Colouring} 
is \NP-complete. 

Due to its computational hardness, the {\sc Colouring} problem has been well studied for special graph classes. We refer to
the survey~\cite{GJPS17} for an overview of the results on {\sc Colouring} restricted to graph classes characterized by one or two forbidden induced subgraphs. In particular, Kr\'al', Kratochv\'{\i}l, Tuza, and Woeginger \cite{KKTW01} classified {\sc Colouring} for {\it $H$-free graphs}, that is, graphs that do not contain a single graph $H$ as an induced subgraph. To explain their result we need the following notation.
For a graph $F$, we write $F\ssi G$ to denote that $F$ is an induced subgraph of a graph $G$.
The {\it disjoint union} of two graphs $G_1$ and $G_2$ is the graph $G_1+G_2$, which has vertex set $V(G)\cup V(H)$ and edge set $E(G)\cup E(H)$. We write $rG$ for the disjoint union of $r$ copies of $G$ by $rG$. The graphs $P_r$ and $C_r$ denote the induced path and cycle on $r$ vertices, respectively.
We can now state the theorem of Kr\'al et al.

\begin{theorem}[\cite{KKTW01}]\label{t-dicho}
Let $H$ be a graph.
If $H\ssi  P_4$ or $H\ssi P_1+ P_3$, then
{\sc Coloring} restricted to $H$-free graphs is polynomial-time solvable, otherwise it is \NP-complete. 
\end{theorem}

For a vertex $u$ or edge $e$ in a graph $G$, we let $G-u$ and $G-e$ be the graph obtained from $G$ by deleting $u$ or $e$, respectively.
Note that such an operation may reduce the chromatic number of the graph by at most~1.
We say that $u$ or $e$ is {\it critical} if $\chi(G-u)=\chi(G)-1$ or $\chi(G-e)=\chi(G)-1$, respectively.  
A graph is vertex-critical if every vertex is critical and edge-critical if every edge is critical.
To increase our understanding of the {\sc Colouring} problem and to obtain certifying algorithms that solve {\sc Colouring} for special graph classes, vertex-critical and edge-critical graphs have been studied intensively in the literature, see for
instance~\cite{BHS09,CGSZ,CGSZ15,DHHMMP,GS15,HH13,HMRSV15,MM12} for certifying algorithms for (subclasses of) $H$-free graphs and in particular $P_r$-free graphs.

In this paper we consider the problems {\sc Critical Vertex} and {\sc Critical Edge}, which are to test if a graph has a critical vertex or critical edge, respectively.
In addition we also consider the edge contraction variant of these two problems. 
We let $G/e$ denote the graph obtained from $G$ after contracting $e=vw$, that is, after removing $v$ and $w$ and replacing them by a new vertex made adjacent to precisely those vertices adjacent to $v$ or $w$ in~$G$ (without creating multiple edges). Contracting an edge may reduce the chromatic number of the graph by at most~1.
An edge~$e$ is {\it contraction-critical} if $\chi(G/e)=\chi(G)-1$. This leads to the {\sc Contraction-Critical Edge} problem, which is to test if a graph has a contraction-critical edge.

\subsection{Our Results}

We prove the following complexity dichotomies for {\sc Critical Vertex}, {\sc Critical Edge} and {\sc Contraction-Critical Edge} restricted to $H$-free graphs. 

\begin{theorem}\label{t-critical}
If a graph $H\ssi  P_4$ or of $H\ssi P_1+ P_3$, then
{\sc Critical Vertex}, {\sc Critical Edge} and {\sc Contraction-Critical Edge} restricted to $H$-free graphs are polynomial-time solvable, otherwise they are \NP-hard or co-\NP-hard.
\end{theorem}

We note that the classification in Theorem~\ref{t-critical} coincides with the one in Theorem~\ref{t-dicho}. The polynomial-time cases for {\sc Critical Vertex} and {\sc Contraction-Critical Edge} can be obtained from Theorem~\ref{t-dicho}.  The reason for this is that a class of $H$-free graphs is not only closed under vertex deletions, but also under edge contractions whenever $H$ is a 
{\it linear forest}, that is, a disjoint union of a set of paths (see Section~\ref{s-main} for further details). 
However, no class of $H$-free graphs is closed under edge deletion. We get around this issue by proving, in Section~\ref{s-equi}, that an edge is critical if and only if it is contraction-critical. Hence, {\sc Critical Edge} and {\sc Contraction-Critical Edge} are equivalent.

The \NP-hardness constructions of Theorem~\ref{t-dicho} cannot be used for proving the hard cases for  {\sc Critical Vertex}, {\sc Critical Edge} and {\sc Contraction-Critical Edge}. 
Instead we construct new hardness reductions in Sections~\ref{s-clawcycle} and~\ref{s-linearforest}.
In Section~\ref{s-clawcycle} we prove that the three problems are \NP-hard for $H$-free graphs if  $H$ contains a claw or a cycle on three or more vertices. 
In the remaining case $H$ is a linear forest. In Section~\ref{s-linearforest} we prove that the three problems are co-\NP-hard even for  $(C_5,4P_1,2P_1+P_2, 2P_2)$-free graphs. In Section~\ref{s-main} we combine the known cases with our new results from Sections~\ref{s-equi}--\ref{s-linearforest} in order to prove Theorem~\ref{t-critical}.

\subsection{Consequences}

Our results have consequences for the computational complexity of two graph blocker problems. Let $S$ be some fixed set of graph operations, and let $\pi$ be some fixed graph parameter. Then, for a given graph $G$ and integer $k\geq 0$, the {\sc $S$-Blocker($\pi$)} problem asks if $G$ can be modified into a graph $G'$ by using at most $k$ operations from $S$ so that $\pi(G')\leq \pi(G)-d$ for some given {\it threshold} $d\geq 0$.    
Over the last few years, the {\sc $S$-Blocker($\pi$)} problem has been well studied, see for instance~\cite{BBPR,BTT11,Bentz,CWP11,DPPR15,PBP,PPR17b,PPR16,PPR17,RBPDCZ10}. 
If $S$ consists of a single operation that is either a vertex deletion or edge contraction, then $S$-{\sc Blocker($\pi$)} is called {\sc Vertex Deletion Blocker($\pi$)} or {\sc Contraction Blocker($\pi$)}, respectively. 
By taking $d=k=1$ and $\pi=\chi$ we obtain the problems {\sc Critical Vertex} and {\sc Contraction-Critical Edge}, respectively. 
We showed in~\cite{PPR17b} how the results for {\sc Critical Vertex} and {\sc Contraction-Critical Edge} can be extended with other results to get complexity dichotomies for {\sc Vertex Deletion Blocker($\chi$)} and {\sc Contraction Blocker($\chi$)} for $H$-free graphs.  

\subsection{Future Work}

A graph $G$ is {\it $(H_1,\ldots,H_p)$-free} for some family of graphs $\{H_1,\ldots,H_p\}$ and integer $p\geq 2$ if $G$ is $H$-free for every $H\in \{H_1,\ldots,H_p\}$.
As a direction for future research we propose classifying the computational complexity of our three problems for
$(H_1,\ldots,H_p)$-free graphs for any $p\geq 2$.
We note that such a classification for {\sc Coloring} is still wide open even for $p=2$ (see~\cite{GJPS17}).
Hence, research in this direction might lead to an increased understanding of the complexity of the {\sc Coloring} problem.

\section{Equivalence}\label{s-equi}

We prove the following result, which implies that the problems {\sc Critical Edge} and {\sc Contraction-Critical Edge} are equivalent.

\begin{proposition}\label{p-equivalent}
An edge is critical if and only if it is contraction-critical.
\end{proposition}

\begin{proof}
Let $e=uv$ be an edge in a graph $G$. First suppose that $e$ is critical, so $\chi(G-e)=\chi(G)-1$. Then $u$ and $v$ are colored alike in any coloring of $G-e$
that uses $\chi(G-e)$ colors. Hence, the graph $G/e$ obtained from contracting $e$ in $G$ can also be colored with $\chi(G-e)$ colors. Indeed, we simply copy a $(\chi(G-e))$-coloring of $G-e$ such that the new vertex in $G/e$ is colored with the same color as $u$ and $v$ in $G-e$. 
Hence $\chi(G/e)=\chi(G-e)=\chi(G)-1$, which means that $e$ is contraction-critical.

Now suppose that $e$ is contraction-critical, so $\chi(G/e)=\chi(G)-1$. By copying a $\chi(G/e)$-coloring of $G/e$ such that $u$ and $v$ are colored with the same color as the new vertex in $G/e$, we obtain a coloring of $G-e$. So we can color $G-e$ with $\chi(G/e)$ colors as well.
Hence $\chi(G-e)=\chi(G/e)=\chi(G)-1$, which means that $e$ is critical.\qed
\end{proof}

\section{Forbidding Claws or Cycles}\label{s-clawcycle}

The \emph{claw} is the 4-vertex star $K_{1,3}$ on vertices $a,b,c,d$ and edges $ab$, $ac$ and $ad$.
In this section we prove that the problems {\sc Critical Vertex}, {\sc Critical Edge} and {\sc Contraction-Critical Edge} are \NP-hard for $H$-free graphs whenever the graph~$H$ contains 
a claw or a cycle on at least three vertices.

Let ${\cal G}$ be a graph class with the following property: if $G\in {\cal G}$, then so are $2G$ and $G+K_r$ for any $r\geq 1$. We call such a graph class {\it clique-proof}.

\begin{theorem}\label{t-col}
If {\sc Coloring}  is \NP-complete for a clique-proof graph class~${\cal G}$, then both {\sc Critical Vertex} and
{\sc Contraction-Critical Edge} are \NP-hard for ${\cal G}$.
\end{theorem}

\begin{proof}
Let ${\cal G}$ be a graph class that is clique-proof. From a given graph $G\in {\cal G}$ and integer $\ell\geq 1$ we construct the graph $G'=2G+K_{\ell+1}$. Note that $G'\in {\cal G}$ by definition and that $\chi(G')=\max\{\chi(G),\ell+1\}$. 
We first prove that $\chi(G)\leq \ell$ if and only if $G'$ contains a contraction-critical edge. 

Suppose that  $\chi(G)\leq \ell$. Then $\chi(G')=\chi(K_{\ell+1})=\ell+1$. In $G'$ we contract an edge of the $K_{\ell+1}$. This yields the graph $G^*=2G+K_\ell$, which has
chromatic number $\chi(G^*)=\ell$, as $\chi(K_\ell)=\ell$ and $\chi(G)\leq \ell$. As $\chi(G')=\ell+1$, this means that  $\chi(G^*)= \chi(G')-1$. Hence $G'$ contains a contraction-critical edge.

Now suppose that $G'$ contains a contraction-critical edge. Let $G^*$ be the resulting graph after contracting this edge.
Then $\chi(G^*)=\chi(G')-1$. As contracting an edge in one of the two copies of $G$ in $G'$ does not lower the chromatic number of $G'$, the contracted edge must be in the $K_{\ell+1}$, that is, $G^*=2G+ K_\ell$.
As this did result in a lower chromatic number, we conclude that $\chi(G')=\chi(K_{\ell+1})=\ell+1$ and $\chi(G^*)=\chi(2G+ K_\ell)=
\max\{\chi(G),\ell\}=\ell$.
The latter equality implies that $\chi(G)\leq \ell$.

From the above we conclude that {\sc Contraction-Critical Edge} is \NP-hard.
We can prove that {\sc Critical Vertex} is \NP-hard by using the same arguments.\qed
\end{proof}

We also need a result of Maffray and Preissmann as a lemma.

\begin{lemma}[\cite{MP96}]\label{l-3col}
The {\sc $3$-Coloring} problem  is \NP-complete for $C_3$-free graphs.
\end{lemma}

\begin{figure}
\begin{center}
\includegraphics[scale=1]{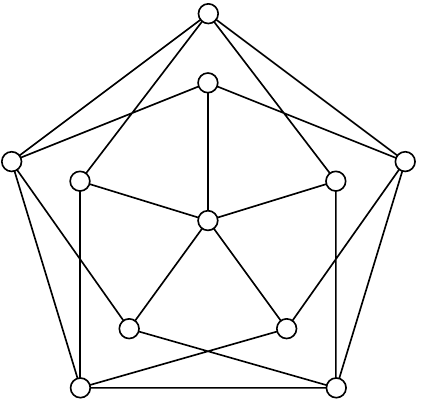}
\caption{The Gr\"otzsch graph.}
\label{f-gr}
\end{center}
\end{figure}

We are now ready to prove the main result of this section.

\begin{theorem}\label{t-known}
Let $H$ be a graph such that  $H\si K_{1,3}$ or $H\si C_r$ for some $r\geq 3$. Then the problems {\sc Critical Vertex}, {\sc Critical Edge} and {\sc Contraction-Critical Edge} are \NP-hard for $H$-free graphs.
\end{theorem}

\begin{proof}
By Proposition~\ref{p-equivalent} it suffices to consider {\sc Critical Vertex} and {\sc Contraction-Critical Edge}.
If $H$ is not a clique, then the class of $H$-free graphs is clique-proof. Hence, in this case, we can use Theorems \ref{t-dicho} and ~\ref{t-col} to obtain \NP-hardness. 

Suppose $H$ is a clique. 
It suffices to 
show \NP-completeness for $H=C_3$. We reduce from 3-{\sc Coloring} restricted to $C_3$-free graphs. This problem is \NP-complete by Lemma~\ref{l-3col}. 
Let $G$ be a $C_3$-free graph that is an instance of {\sc $3$-Coloring}. We obtain an instance of {\sc Critical Vertex} or
{\sc Contraction-Critical Edge} as follows. Take the disjoint union of two copies of $G$ and the
Gr\"otzsch graph~$F$ (see Figure~\ref{f-gr}), which is known to be 4-colorable but not 3-colorable (see~\cite{West}). Call the resulting graph~$G'$, so $G'=2G+F$. As $G$ and $F$ are $C_3$-free, $G'$ is $C_3$-free. We claim that $G$ is 3-colorable if and only if $G'$ has a critical vertex if and only if $G'$ has a contraction-critical edge. This can be proven via similar arguments as used in the proof of Theorem~\ref{t-col}, with $F$ playing the role of $K_{\ell+1}$.
\qed
\end{proof}

Note that in Theorem~\ref{t-known} we cannot prove membership in \NP, as {\sc Coloring} is \NP-complete for the class of $H$-free graphs if $H\si K_{1,3}$ or $H\si C_r$ for some $r\geq 3$ due to Theorem~\ref{t-dicho}. As such, it is not clear if there exists a certificate.

\section{Forbidding Linear Forests}\label{s-linearforest}

In this section we prove our second hardness result needed to show Theorem~\ref{t-critical}. We first introduce some additional terminology.

Let $G$ be a graph.
The graph $\overline{G}$ denotes the {\it complement} of $G$, that is, the graph with vertex set $V(G)$ and an edge between two vertices $u$ and $v$ if and only if $u$ and~$v$ are not adjacent in $G$. 
A subset $K$ of vertices in $G$ is a {\it clique} if any two vertices in~$K$ are adjacent to each other.
A {\it clique cover} of a graph~$G$ is a set ${\cal K}$ of cliques in~$G$, such that each vertex of~$G$ belongs to exactly one clique of ${\cal K}$.
The {\it clique covering number}~$\sigma(G)$ is the size of a smallest clique cover of $G$. Note that $\chi(G)=\sigma(\overline{G})$. The size of a largest clique in a graph $G$ is denoted by $\omega(G)$.

The hardness construction in the proof of our next result uses clique covers. Kr\'al et al.~\cite{KKTW01} proved
that {\sc Coloring} is \NP-hard for  $(C_5,4P_1,P_1+2P_2, 2P_2)$-free graphs.  
This does not give us hardness for {\sc Critical Vertex} or {\sc Critical Edge}, but we can use some elements of their construction.
For instance, we reduce from a similar \NP-complete problem as they do, namely the \NP-complete problem {\sc Monotone 1-in-3-SAT}, which is defined as follows.
Let $\Phi$ be a formula with clause set~$C$ of size~$m$ and variable set~$X$ of size $n$, so that each clause in $C$ consists of three distinct positive literals, and each variable in $X$ occurs in exactly three clauses. The question is whether $\Phi$ has a truth assignment, such that each clause is satisfied by exactly one variable. In that case we say that $\Phi$ is {\it 1-satisfiable}.
Note that $m=n$. 
Moore and Robson proved that this problem is \NP-complete.

\begin{lemma}[\cite{MR01}]\label{l-np}
{\sc Monotone 1-in-3-SAT} is \NP-complete.
\end{lemma}

We are now ready to prove the main result of this section.

\begin{theorem}\label{t-main}
The problems {\sc Critical Vertex}, {\sc Critical Edge} and {\sc Contraction-Critical Edge} are co-\NP-hard for $(C_5,4P_1,2P_1+P_2, 2P_2)$-free graphs.
\end{theorem}

\begin{proof}
By Proposition~\ref{p-equivalent} it suffices to consider {\sc Critical Vertex} and {\sc Critical Edge}.
We will first consider {\sc Critical Vertex} and show that the equivalent problem whether a graph has a vertex whose deletion reduces
the clique covering number by~1 is co-\NP-hard for $(C_4,C_5,K_4,\overline{2P_1+P_2})$-free graphs. 
We call such a vertex {\it critical} as well. 
The complement of a  $(C_4,C_5,K_4,\overline{2P_1+P_2})$-free graphs is $(C_5,4P_1,2P_1+P_2, 2P_2)$-free.
Hence by proving this co-\NP-hardness result we will have proven the theorem for {\sc Critical Vertex}.

As mentioned, we reduce from {\sc Monotone 1-in-3-SAT}, which is \NP-complete due to Lemma~\ref{l-np}.
Given an instance $\Phi$ of  {\sc Monotone 1-in-3-SAT} with clause set~$C$ and variable set $X$, we construct a graph $G=(V,E)$ as follows. For every clause $c\in C$, the clause gadget $G_c=(V_c,E_c)$ is a cycle of length~$7$. For $c=(x,y,z)$, we let three pairwise non-adjacent vertices $c(x),c(y),c(z)$ of $G_c$ correspond to the three variables $x,y,z$. 
We denote the other four vertices of $G_c$ by $a_i^c,1\le i\le 4$, so that $G_c=c(x)a_1^ca_2^cc(y)a_3^cc(z)a_4^cc(x)$.
For each variable $x\in X$ we let the variable gadget $Q_x$ consist of the triangle $c(x)c'(x)c''(x)c(x)$, where $c,c',c''$ are the three clauses containing~$x$. See Figure~\ref{fig:clause} for an illustration of the construction.
We observe that $\vert{V(G)}\vert=7n$ and that $G$ is $(C_4,C_5,K_4,\overline{2P_1+P_2})$-free with $\omega(G)=3$.

\begin{figure}[ht!]
\begin{center}
\includegraphics[keepaspectratio=true, width=8cm]{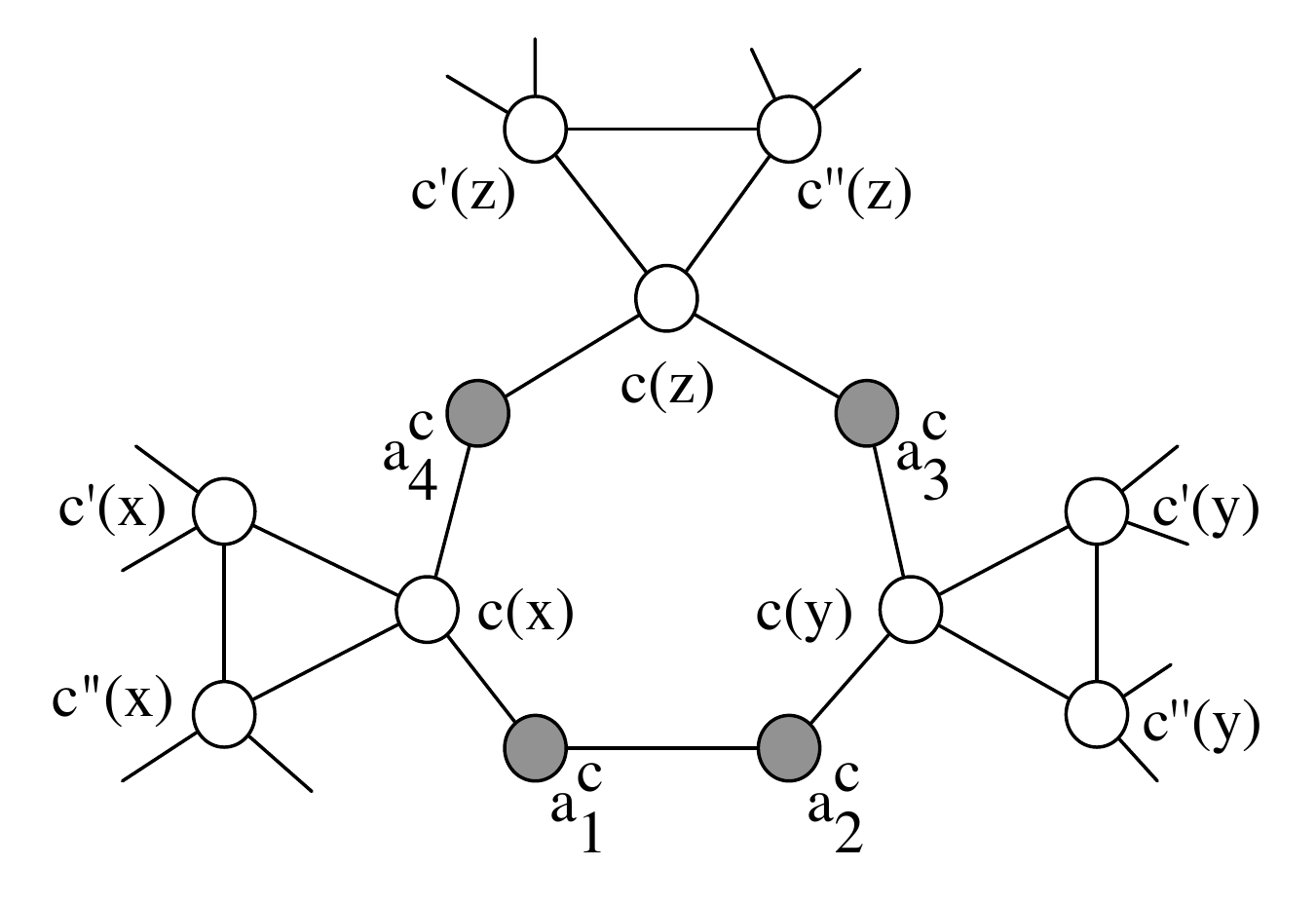}
\caption{The clause gadget $G_c$ and three variable gadgets $Q_x$, $Q_y$ and $Q_z$.}\label{fig:clause}
\end{center}
\end{figure}
\noindent

In order to prove co-\NP-hardness we first need to deduce a number of properties of our gadget. We do this via a number of claims.

\medskip
\noindent
{\it Claim~1.} 
There exists a minimum clique cover of $G$, in which each $a_i^c$ is covered by a clique of size~2, and moreover, every two vertices $a_1^c$ and $a_2^c$ belong to the same (2-vertex) clique.

\medskip
\noindent
We prove Claim~1 as follows. Let ${\cal K}$ be a minimum clique cover of $G$. Suppose two vertices $a_1^c$ and $a_2^c$ belong to two different cliques $K$, $K'$ of ${\cal K}$. 
If one of $K$, $K'$ has size~1, say $K=\{a^1_c\}$, then we can replace $K$ and $K'$ by $\{a_1^c,a_2^c\}$ and $K'\setminus \{a_2^c\}$. This yields a new minimum size clique cover of $G$, in which $a_1^c$ and $a_2^c$ belong to the same clique.
Alternatively, if $K$ and $K'$ each have size~2, then $K=\{a_1^c,c(x)\}$ and $K'=\{a_2^c,c(y)\}$. Then, by construction, ${\cal K}$ contains 
a clique that either consists of $a_3^c$ or of $a_4^c$ , say $a_4^c$. We replace the cliques $\{a_4^c\}$ and $K$ by $\{a_4^c,c(x)\}$ and $\{a_1^c\}$, respectively, and return to the previous situation.
Hence we may assume without loss of generality that $\{a_1^c,a_2^c\}$ is a clique in~$K$.
This means that if $a_3^c$ or $a_4^c$ forms a 1-vertex clique in ${\cal K}$, then we can safely add $c(y)$ or $c(x)$, respectively, to it. This proves Claim~1.

\medskip
\noindent
Now let ${\cal K}$ be a minimum clique cover. By Claim~1, we may assume without loss of generality that each $a_i^c$ is covered by a clique of size~2, and moreover, that every two vertices $a_1^c$ and $a_2^c$ belong to the same (2-vertex) clique.
Since the clause gadgets~$G_c$ are pairwise non-intersecting and isomorphic to $C_7$, it takes at least four cliques to cover the vertices of every~$G_c$. 
This means that exactly $3n$ cliques are needed to cover the $4n$ vertices $a_i^c$. By construction, we also find that $2n$ vertices $c(x)$ are covered by these cliques. Since $\omega(G)=3$, at least $n/3$ other cliques are necessary to cover the $n$ remaining vertices~$c(x)$. Hence, ${\cal K}$ has size at least ${10\over 3}n$, that is, 
$$\sigma({G})\ge{10\over 3}n.$$

We now prove three more claims.

\medskip
\noindent
{\it Claim~2.}
$\Phi$ is 1-satisfiable if and only if $\sigma({G})={10\over 3}n$.

\medskip
\noindent
We prove Claim~2 as follows. First suppose $\Phi$ is 1-satisfiable. We construct a clique cover~${\cal K}$ in the following way.
 If $x$ is true, then we let ${\cal K}$ contain the triangle $Q_x$.  Since each clause $c$ contains exactly one true variable for each $G_c$, exactly one vertex of $G_c$ is covered by a variable gadget. Then ${\cal K}$ contains three cliques of size~2 covering the six other vertices of~$G_c$. 
Hence ${\cal K}$ has size ${10\over 3}n$. As $\sigma({G})\ge{10\over 3}n$, this implies that
$\sigma({G})={10\over 3}n$.

Now suppose $\sigma({G})={10\over 3}n$. 
Let ${\cal K}$ be a minimum clique cover of~$G$.
 By Claim~1, we may assume without loss of generality that each $a_i^c$ is covered by a clique of size~2, and moreover, that every two vertices $a_1^c$ and $a_2^c$ belong to the same (2-vertex) clique. Then at least $n/3$ other cliques are necessary to cover the vertices $c(x)$ that are not in a 2-vertex clique with a vertex $a_i^c$.
Hence, as  $\sigma({G})={10\over 3}n$,
these vertices are covered by exactly $n/3$ triangles, each one corresponding to one variable $x$ (these are the only triangles in $G$). We assign the value true to a variable $x\in X$ if and only if its corresponding triangle $Q_x$ is in the clique cover. Then, for each $c\in C$, exactly one variable is true, namely the one that corresponds to the unique vertex of $G_c$ covered by a triangle. So $\Phi$ is 1-satisfiable. This completes the proof of Claim~2.

\medskip
\noindent
{\it Claim~3.}
If $G$ has a clique cover ${\cal K}=\{K_1,\ldots,K_{{10\over 3}n}\}$, then each $K_i\in\cal K$ 
consists of either two or three vertices.

\medskip
\noindent
We prove Claim~3 as follows.
As $\sigma({G})\ge{10\over 3}n$ and $\vert{\cal K}\vert={10\over 3}n$, we find that ${\cal K}$ is a minimum clique cover.
With each $v\in V$, we associate a {\it weight} $w_v\ge 0$ as follows. For $K_i\in \cal K$ and $v\in K_i$, we define $w_v={1/\vert K_i\vert}$. Since $\omega({G})= 3$ we have $w_v\in\{{1\over 3},{1\over 2},1\}$. So we have 
\[\displaystyle\sum_{{G_c}}\sum_{v\in V_c}w_v=\sum_{v\in V}w_v= \sum_{i=1}^{{10\over 3}n}\sum_{v\in K_i}w_v={10\over 3}n, \] 
where the first equality holds, because the clause gadgets $G_c$ are vertex-disjoint. We show that for every $c$ we have $\Sigma_{v\in V_c}w_v\ge{10\over 3}$. Since every $a_i^c$ has exactly two neighbours and these neighbours are not adjacent, we have $w_{a_i^c}\in \{{1\over 2},1\}$. If there exists an index~$i$ such that $w_{a_i^c}=1$, then 
\[\displaystyle\sum_{v\in V_c}w_v\ge1+3\times{1\over 2}+3\times{1\over 3}={7\over 2}>{10\over 3}.\] Now if $a_i^c$ has weight $w_{a_i^c}={1\over 2}$ for each $1\leq i\leq 4$, then $a_i^c$ is covered by a clique of size~$2$ and the second vertex of this clique has weight ${1\over 2}$ as well by definition. Thus if $w_{a_i^c}={1\over 2}$, exactly two among $c(x),c(y),c(z)$ have weight $1\over 2$. It follows that 
\[ \displaystyle\sum_{v\in V_c}w_v\ge 4\times{1\over 2}+2\times{1\over 2}+{1\over 3}={10\over 3}.\]
Hence $\sum_{v\in V_c}w_v={10\over 3}$ if and only if each vertex of $G_c$ is in a clique of size 2 or 3.
Since $\sum_{{G_c}}\sum_{v\in V_c}w_v={10\over 3}n$, we obtain $\sum_{v\in V_c}w_v={10\over 3}$ for every $c\in C$.
We conclude that each clique in ${\cal K}$ is of size $2$ or $3$. 
This completes the proof of Claim~3.

\medskip
\noindent
{\it Claim~4.}
If $\sigma(G)>{10\over 3}n$, then $G$ has a minimum clique cover ${\cal K}$ that contains a clique of size~1.

\medskip
\noindent
We prove Claim~4 as follows.
Suppose $\sigma(G)>{10\over 3}n$. 
For contradiction, assume that every minimum clique cover of $G$ has no clique of size~1.
Let ${\cal K}$ be a minimum clique cover of $G$.
By Claim~1, we may assume without loss of generality that each $a_i^c$ is covered by a clique of size~2, and moreover, that every two vertices $a_1^c$ and $a_2^c$ belong to the same (2-vertex) clique. Hence the remaining vertex $c(x)$ is covered by  some clique $K_i\in {\cal K}$,  such that either $K_i=\{c(x),c'(x)\}$ or  $K_i=\{c(x),c'(x),c''(x)\}$. 

If $K_i=\{c(x),c'(x)\}$, then $c''(x)$ is covered by some clique $K_j=\{c''(x),a\}$. However, then we can take $K_i=\{c(x),c'(x),c''(x)\}$ and $K_j=\{a\}$ to obtain a minimum clique cover with $\vert K_j\vert=1$, a contradiction. Hence $K_i=\{c(x),c'(x),c''(x)\}$. As this holds for every $G_c$ we find that $\sigma(G)={10\over 3}n$, a contradiction.
This completes the proof of Claim~4.

\medskip
\noindent
We claim that $\Phi$ is a 1-satisfiable if and only if $G$ has no critical vertex.
First suppose that $\Phi$ is 1-satisfiable.
By Claims~2 and~3 we find that $\sigma(G)={10\over 3}n$ and every clique in any minimum clique cover of $G$ has size greater than $1$. Hence, there is no vertex $u$ of~$G$ with $\sigma(G-u)\leq \sigma(G)-1$, that is, $G$ has no critical vertex.

Now suppose that $\Phi$ is not 1-satisfiable. By Claims~2 and~4 
we find that $\sigma(G)>{10\over 3}n$ and that there exists a minimum clique cover that contains a clique~$\{u\}$ of size~1. This 
means that $\sigma(G-u)=\sigma(G)-1$. So $u$ is a critical vertex.

\medskip
\noindent
We are left to consider the {\sc Critical Edge} problem. We use the same construction as before except that the cycles $G_c$ are isomorphic to $C_{11}$. To be more precise, we let $G_c=c(x)a_1^ca_2^cc(y)a_3^ca_4^ca_5^cc(z)a_6^ca_7^ca_8^cc(x)$.
Again the resulting graph~$G$ is $(C_4,C_5,K_4,\overline{2P_1+P_2})$-free.
 By using the same arguments as before we find that if $\Phi$ is 1-satisfiable, then every clique in any minimum clique cover of $G$ has size greater than $1$. Hence, as $G$ is $K_4$-free, every clique in any minimum clique cover of $G$ has size~2 or~3. 
Since $G$ is $\overline{2P_1+P_2}$-free, we cannot merge two cliques into one by adding a new edge.
So $\overline{G}$ has no critical edge. 

Now suppose that $\Phi$ is not 1-satisfiable. Then using the previous arguments we can prove that there exists a minimum clique cover~${\cal K}$ that contains a clique $\{u\}$ of size~1. By the adjusted construction of $G_c$ we find that $u$ is adjacent to exactly one vertex of a 2-vertex clique $\{v,w\}$ of ${\cal K}$, say $u$ is adjacent to $v$ but not to $w$. Then by adding the edge $uw$, 
which yields the graph~$G+uw$, 
we merge two cliques into one, meaning that
$\sigma(G+uw)=\sigma(G)-1$. So~$uv$ is a critical edge of~$\overline{G}$. This completes the proof of Theorem~\ref{t-main}.\qed
\end{proof}

\section{The Proof of Theorem~\ref{t-critical}}\label{s-main}

We are now ready to prove Theorem~\ref{t-critical}, which we restate below.

\medskip
\noindent
{\bf Theorem~\ref{t-critical}.}
{\it If a graph $H\ssi  P_4$ or of $H\ssi P_1+ P_3$, then
{\sc Critical Vertex}, {\sc Critical Edge} and {\sc Contraction-Critical Edge} restricted to $H$-free graphs are polynomial-time solvable, otherwise they are \NP-hard or co-\NP-hard.}

\begin{proof}
Let $H\ssi P_1+P_3$ or $H\ssi P_4$. Let $G$ be an $H$-free graph. 
By Theorem~\ref{t-dicho} we can compute $\chi(G)$ in polynomial time. We note that any vertex deletion  results in a graph that is $H$-free as well. Hence in order to solve {\sc Critical Vertex} we can compute the chromatic number of
$G-v$ for each vertex $v$ in polynomial time and compare it with $\chi(G)$.
As $(P_1+P_3)$-free graphs and $P_4$-free graphs are closed under edge contraction as well, we can follow the same approach for solving {\sc Contraction-Critical Edge}. 
By Proposition~\ref{p-equivalent} we obtain the same result for {\sc Critical Edge}.

Now suppose that neither $H\ssi P_1+P_3$ nor $H\ssi P_4$. If $H$ has a cycle or an induced claw, then we use 
Theorem~\ref{t-known}. Assume not. Then $H$ is a disjoint union of $r$ paths for some $r\geq 1$. If $r\geq 4$ we use Theorem~\ref{t-main}. If $r=3$ then either $H=3P_1 \ssi P_1+P_3$, which is not possible, or
$H\si 2P_1+P_2$ and we can apply Theorem~\ref{t-main} again. 
Suppose $r=2$. If both paths contain an edge, then $2P_2\ssi H$. If at most one path has edges, then it must have
at least four vertices, as otherwise $H\ssi P_1+P_3$. This means that $2P_1+P_2\ssi H$. In both cases we apply Theorem~\ref{t-main}.
If $r=1$, then $H$ is a path on at least five vertices, which means $2P_2\ssi H$. We
apply Theorem~\ref{t-main} again.\qed
\end{proof}

\end{document}